\documentclass[11pt,letterpaper]{article}
\usepackage{fullpage,xspace,times}
\usepackage{graphicx}
\usepackage{amsmath,upgreek}
\usepackage{amsfonts}
\usepackage{amssymb}
\usepackage{epstopdf}

\newtheorem{theorem}{Theorem}

\newtheorem{assumption}[theorem]{Assumption}

\newtheorem{lemma}[theorem]{Lemma}

\newcommand{\cs}{{\cal{S}}}

\newcommand{\Alpha}{\mbox{\textsc{Asym}}}
\newcommand{\Beta}{\mbox{\textsc{Abs}}}

\newcommand{\Xomit}[1]{ }
\newenvironment{proof}[1][Proof]{\textbf{#1.} }{\ \rule{0.5em}{0.5em}}

\mathchardef\mhyphen="2D

\newcommand{\eps}{\upvarepsilon}

\newcommand{\mmu}{\upmu}

\begin{document}
\begin{titlepage}

\title{Truly asymptotic lower bounds for online vector bin packing}

\date{}

\author{J\'anos Balogh\thanks{Institute of Informatics,
     University of Szeged, Szeged, Hungary. \texttt{baloghj@inf.u-szeged.hu}. }   \and Leah
Epstein\thanks{ Department of Mathematics, University of Haifa,
Haifa, Israel. \texttt{lea@math.haifa.ac.il}. } \and Asaf
Levin\thanks{Faculty of Industrial Engineering and Management, The
Technion, Haifa, Israel. \texttt{levinas@ie.technion.ac.il.}}}

%\vspace{-0.5cm}

\maketitle

\thispagestyle{empty}

\begin{abstract}
In this work, we consider online vector bin packing.  It is known that no algorithm can have a competitive ratio of $o(d/\log^2 d)$ in the absolute sense, though upper bounds for this problem were always shown in the asymptotic sense. Since variants of bin packing are traditionally studied with respect to the asymptotic measure and since the two measures are different, we focus on the asymptotic measure and prove new lower bounds on the asymptotic competitive ratio. The existing lower bounds prior to this work were much smaller than $3$ even for very large dimensions.

We significantly improve the best
known lower bounds on the asymptotic competitive ratio (and as a byproduct, on the absolute competitive ratio) for online vector packing of vectors with $d \geq 3$ dimensions, for every such dimension $d$.  To obtain these results, we use several different constructions, one of which is an adaptive construction showing a lower bound of $\Omega(\sqrt{d})$.
Our main result is that the lower bound of $\Omega(d/\log^2 d)$ on the competitive ratio holds also in the asymptotic sense. The last result requires a careful adaptation of constructions for online coloring rather than simple black-box reductions.
\end{abstract}
\end{titlepage}

\section{Introduction}
We study the vector packing problem (VP) \cite{GareyGJ76,GaKeWo94,ACKS13,ACFR16,AzarCR16}. In
VP with dimension $d \geq 2$, a set of items is given, where every item is
a non-zero $d$-dimensional vector whose components are rational
numbers in $[0,1]$. This set  is to be split into subsets called bins, such that the vector sum of every subset does
not exceed $1$ in any component. This constraint is called the capacity constraint. Cardinality Constrained Bin Packing (CCBP) \cite{BCKK04,BBDEL_CCBP} is a special case of VP with dimension $2$, defined as follows. Seeing CCBP as a one-dimensional packing problem, there is an integer parameter $k \geq 2$  such that in addition to the capacity constraint, no bin can have more than $k$ items. In an equivalent input
for VP with $d\geq2$,  the second
component of every item is $\frac 1k$, while the first component is in $[0,1]$ (the other components are zeroes).

In this paper we consider lower bounds on the worst-case performance guarantees of online algorithms for VP.
Online algorithms  receive input
items one by one, and pack each new item irrevocably before the
next item is presented, into an empty (new) bin or non-empty bin.
Such algorithms receive an input as a sequence, while offline
algorithms receive an input as a set. An arbitrary
optimal offline
algorithm, which uses the minimum number of bins for packing the items
of the input or instance,
is denoted by $OPT$. For an input $L$ and algorithm $A$, we
let $A(L)$ denote the number of bins that $A$ uses to pack $L$ also called the {\em cost of algorithm $A$ on input $L$}. We
also let $OPT(L)$ denote the number of bins that $OPT$ uses for
a given input $L$.  The absolute competitive ratio of an
algorithm $A$ is defined as the supremum ratio over the following
ratios. These are the ratios for all inputs $L$, where the ratio
for $L$ is between the number of its bins $A(L)$ and the number of
the bins of $OPT$, $OPT(L)$. The asymptotic competitive ratio is
the limit of absolute competitive
ratios $R_K$ when $K$ tends to
infinity, and $R_K$ takes into account only inputs for which $OPT$
uses at least $K$ bins. That is, the asymptotic competitive ratio
of $A$ is
$$\lim_{K \rightarrow \infty} \   \sup_{OPT(L)\geq K}
\  \frac{A(L)}{OPT(L)} \ . $$  In this paper we mostly deal with the
asymptotic competitive ratio, which is the most natural measure for bin packing algorithms, and we sometimes refer to it by the term
competitive ratio. When we discuss the absolute competitive ratio,
we use this last term explicitly.

Note that VP is defined as a distinct optimization problem for every fixed value of $d$.
For each such
value (which is a dimension), there might be an online algorithm that is the best possible with respect to the absolute competitive ratio (the one whose absolute competitive ratio is minimized), and there might be an online algorithm that is the best possible with respect to the asymptotic competitive ratio. Such algorithms may be different.

Denote by $\Beta(d)$ the best possible absolute competitive ratio of an online algorithm for VP with dimension $d$, and let $\Alpha(d)$ denote  the best possible asymptotic competitive ratio of an online algorithm for VP with dimension $d$. If there are values of $d$ for which such best possible online algorithms do not exist, we define the corresponding values of $\Beta(d)$ and $\Alpha(d)$ as the infimums of the corresponding ratios where the infimum is taken over all online algorithms for VP with $d$ dimensions.
Thus these values are well-defined for all $d$ even though we currently do not know their values. It is known that these values are at most linear in $d$ (see below).

Observe that both $\Beta(d)$ and $\Alpha(d)$ are monotone non-decreasing functions of $d$ as we can use the online algorithm with larger dimension in order to pack the lower dimension vectors by simulating a higher dimension input using additional components which are always set to $0$.
Furthermore, by the definitions of the asymptotic competitive ratio and the absolute competitive ratio, we conclude that for every $d$, we have $\Alpha(d)\leq \Beta(d)$.

In this work, we improve the known lower bounds on the asymptotic competitive ratio for all fixed values of $d$ that are at least $3$,
for the online VP scenario. This improves the known results for the absolute ratios as well.
That is, we improve upon the state of the art of lower bounding
$\Alpha(d)$ for all $d\geq 3$. Our main result is that the order of growth of the asymptotic competitive ratio is $\Omega( \frac{d}{(\log d)^2})$. Recall that the term asymptotic here is not the asymptotic growth of $d$, but the asymptotic definition of the asymptotic competitive ratio, which is the most common measure for bin packing problems.
To do that, we present four constructions leading to different lower bounds on $\Alpha(d)$.  For every specific value of $d$, one may use the construction that lead to the largest possible lower bound.  In Section \ref{sec:large-values}, we present our results for very large values of $d$.  That is, in that section, we show that for $d > 16$, there is a lower bound of $\frac{d-1}{8(\log_2 d)^3}$ on $\Alpha(d)$, and for sufficiently large and fixed  values of $d$, the lower bound proved for the asymptotic competitive ratio for online VP is $ \frac{d}{2^{11}\cdot (\log_2 d)^2}$.
The statement regarding values of $d$ that are sufficiently large is explained in that section as well.  In Section \ref{sec:medd} we show a lower bound of  $\frac{\lfloor\sqrt{d-1} \rfloor +2}{2}$ that applies for all fixed values of $d \geq 2$, and the same type of construction yields better results for relatively small values of $d$. For example, in the case $d=14$ we get a lower bound of $3$ on the asymptotic competitive ratio.
 Then, in Section  \ref{sec:d=3} we show a lower bound of $\frac{9}{4}=2.25$ on $\Alpha(3)$.
Finally, in Section \ref{sec:d=8}, we show a lower bound of $\frac{76}{29}\approx 2.62$ on $\Alpha(8)$.

The last three constructions, of Sections  \ref{sec:medd},
\ref{sec:d=3}, and \ref{sec:d=8},
are based on using a technique called adaptive constructions for packing problems that we will explain in detail below.  We note that our lower bound of  $\frac{\lfloor\sqrt{d-1} \rfloor +2}{2}$ is the largest lower bound that we prove for a big variety  of values of $d$. This holds even for extremely large values of $d$ like $d=2^{30}$, for which the value of this lower bound is at least $2^{14}> 16000$.
Comparing this result to the values of lower bounds of Section
\ref{sec:large-values}, we find that they are much smaller, and in fact the values resulting from the constructions of that section are not larger than $4972$ and $583$, respectively.  Thus, for every reasonable value of $d$, the lower bounds that we obtain based on the adaptive construction method are much better than the ones of Section \ref{sec:large-values}.

The lower bounds on $\Beta(d)$ established in Azar et al. \cite{ACKS13} were not computed explicitly
in the sense that they
are only stated as $\Omega (d^{1-\eps})$ for every $\eps>0$.  These lower bounds hold only in the absolute sense, and their order of growth is $\Omega\left(\frac{d}{(\log_2 d)^2}\right)$.
In order to compare their bounds to our lower bounds on $\Alpha(d)$, we examined their proofs.  Their lower bounds are basically the deterministic lower bounds on node coloring of graphs where the nodes arrive in an online fashion, thus they are using black-box reductions from the (deterministic) lower bounds of Halld{\'{o}}rsson and Szegedy \cite{HS}.  These black-box reductions are the main reason that they are not able to find lower bounds the asymptotic competitive ratio and only the
absolute competitive ratio which may be larger (and it is not the standard tool to analyze bin packing problems).
Note that we improve the known absolute lower bound for a large number of values of $d$. For example, the deterministic lower bound resulting from the construction of \cite{HS} for $d=60$ is $3.75$ while our results imply a lower bound above $4.8$ (which holds even in asymptotic case and for a slightly lower dimension of $57$).

In the work of \cite{HS}, the part of randomized lower bounds consists of two constructions.  In the first one, the value of $d$ is relatively large so we can use Stirling's formula to approximate $\lceil \log_2 d \rceil !$ with a constant multiplicative error, while the second one holds for all dimensions $d$.  In Section \ref{sec:large-values}, we use the ideas of Halld{\'{o}}rsson and Szegedy \cite{HS} to get similar lower bounds on the asymptotic competitive ratio of VP.  Unlike the work of \cite{ACKS13}, we do not apply black-box reduction from node coloring.  The main motivation for that is that we need a lower bound with respect to a different coloring problem in which we are interested in {\em fractional coloring}.  Since this variant was not studied before, we cannot use an existing result for that, and we cannot use a black-box reduction from it.  Instead, we present the lower bound construction for VP.  The transition to this fractional coloring is the main tool that we use in order to lift the construction of \cite{ACKS13} so that it will allow us to prove a lower bound on the asymptotic competitive ratio for VP.

%%Based on our analysis the transition from the bounds of \cite{ACKS13} to our bounds of Section \ref{sec:large-values} results in an decrease
%%%{increase?} of the multiplicative factor hidden in the $\Omega$ notation by a multiplicative factor of $2$.
%Thus, in the example above where $d=2^{30}$ the lower bounds proved by \cite{ACKS13} on $\Beta(2^{30})$ are smaller than $2\cdot 4972=9944$ that is smaller than the lower bound we prove on $\Alpha(2^{30})$.  Thus, in particular in addition to the major improvement on the lower bounds on $\Alpha(d)$ for all $d\geq 3$ we also improve the
%{constant?}
%lower bounds on the absolute competitive ratio for VP in $d$ dimensions for all reasonable values of $d$ where reasonable here means less than one billion.

Next, we consider the literature of bounds of $\Alpha(d)$.  In \cite{AzarCR16} it is shown that when $d$ grows to infinity there is a lower bound that tends to $e$, thus $\lim_{d\rightarrow \infty} \Alpha(d)\geq e$.  Observe that with respect to this bound we replace $e$ by $\infty$.  The interesting part is that this lower bound of \cite{AzarCR16} applies even if all components of all vectors are small.  That is, for every small $\eps>0$ such that all components are smaller than $\eps$, the lower bound of \cite{AzarCR16} holds, and it holds even when $\eps$ tends to $0$.  This matches the upper bound of $e$ for this very specific case of VP (of small components) established by \cite{ACFR16}.  In \cite{BBDEL_CCBP}, the authors consider the case of $d=2$ and show its difference from CCBP.  With respect to VP in two dimensions they showed that $\Alpha(2)\geq 2.03731129$, and thus for all $d\geq 2$ it was known that $\Alpha(d)\geq 2.03731129$.
Prior to the pair of papers \cite{AzarCR16,BBDEL_CCBP}, the lower bounds on $\Alpha(d)$ were weaker \cite{GaKeWo94,BlVlWo96,Blitz}.
As for upper bounds on $\Alpha(d)$, the best known result \cite{GareyGJ76} is (still) that the First-Fit algorithm has an asymptotic competitive ratio of $d+0.7$ for VP in $d$ dimensions. Prior to that work, there was a slightly weaker bound of $d+1$ by Kou and Markowsky \cite{KouM77}. For the absolute competitive ratio, the resulting bound is also $d+O(1)$ \cite{GareyGJ76, KouM77} (an upper bound of $O(d)$ follows from the simple property that for greedy algorithms, no two bins of the output have a sum of at most $1$ in all components).

We stress the interesting fact that prior to our work it was unknown whether there is an online algorithm $A$ for VP for any dimension $d$ (or for sufficiently large dimensions)  such that for every input $L$, its cost $A(L)$ satisfies $A(L)\leq 3\cdot OPT(L) + 2\cdot d$ (where $d$ is the dimension of the input). The reason for this is that the known lower bound of $e$ holds in fact for any additive term, while the non-constant known lower bound uses just $d$ items for dimension $d$.

We discuss the relation of VP to other bin packing problems. Online CCBP, which is a special case of VP in two dimensions as explained above, is now fully understood. There are matching bounds of $2$ on the asymptotic competitive ratio as well as the absolute competitive ratio (see  \cite{BCKK04,BDE,BBDEL_CCBP}).  As mentioned above, it was already established in the past \cite{BBDEL_CCBP} that VP in two dimensions is strictly harder than CCBP since it is shown that $\Alpha(2)\geq 2.03731129$.
The special case of $d=1$ of VP is simply BP and it was studied for half a century \cite{J74,JDUGG74}.  The current best bounds on $\Alpha(1)$ (for the online version)
are a lower bound of $1.54278$ \cite{BBDEL_lb} and an upper bound of $1.57829$ \cite{BBDEL_ub}, while $\Beta(1)=\frac 53$ \cite{BBDSR}.

There is also vast literature on the offline versions of all the problems mentioned above \cite{CheKha04,BCS09,BEK16,CKP03,EL07afptas,FVL81,KK82}. In particular, similar lower bounds to those of \cite{ACFR16} on the absolute approximation ratio for VP (under a certain standard complexity assumption, that NP$\neq$ZPP) were observed for the offline scenario \cite{CheKha04,BEK16}, also by reductions from coloring problems. The lower bounds are for the absolute measures since for every dimension $d$, the number of items is simply $d$.
%Our methods can also be used to establish hardness of approximation results for the offline variants with respect to the asymptotic measure, while such hardness results were proved only in the absolute sense. This results in asymptotic inapproximability of $\Omega(\frac{d}{\log^2 d})$ unless NP=ZPP.

The other bin packing problems discussed here (BP and CCBP) admit polynomial time asymptotic approximation schemes \cite{CKP03,EL07afptas,FVL81,KK82}, but one can show that unless P=NP, there cannot be such schemes for the absolute cases, and the absolute approximation ratio is at least $1.5$. This can serve as evidence that lower bounds for the absolute approximation ratio or the absolute competitive ratio are not sufficient for the analysis of the asymptotic approximation ratio or the asymptotic competitive ratio. In fact, it holds also for online bin packing problems that the absolute measure is different from the asymptotic one. For example, for BP, we already mentioned that the best possible absolute competitive ratio is $\frac 53$ \cite{BBDSR}, while the asymptotic competitive ratio is much smaller \cite{BBDEL_ub}. For CCBP, one example is the parameter $k=4$, for which the best possible absolute competitive ratio is $2$ \cite{BCKK04,BBDEL_CCBP} while an improved asymptotic competitive ratio below $2$ is known \cite{Epstein05}.

\section{Preliminaries: adaptive constructions}
In some of the constructions, we define inputs using a method
presented in the past \cite{BBDEL_CCBP}. In that method, a binary
condition on the assignment of every item is defined, and it is
used by the adversary (who presents the input) in the definition of the properties of the
following item. More precisely, the algorithm keeps an active
interval of (scalar) values (contained in $(0,1)$), and it modifies the
interval after the assignment of every
input item by the algorithm.
These values are not necessarily the actual sizes of
input items even in
the one-dimensional case, though item sizes are based on them in a
simple pre-specified way.
This means that the generated value is not necessarily the size of the new input item, nor will it always be equal to a component of its vector,
but this generated value is used in the definition of the new item, for example, it may be subtracted from some fixed value or added to a fixed size.

The number of required items is decided in advance, or (in some cases) an upper bound on the number of required items is given in advance in cases where the exact number of required items is revealed later on.  This number of required items is used to decide upon the
initial interval of the values too.
The initial interval is also based on
the required sizes and properties of the values. The initial
interval is always defined such that the smallest size is strictly positive
and the largest size is sufficiently small.

Input items are presented one by one. After the assignment of an item
by the algorithm,
the validity of the condition is tested for this item.  During the
process of input construction, it is ensured (via a process
resembling binary search or geometric binary search) that values
corresponding to items satisfying the binary condition are larger
by a pre-specified (constant)
multiplicative factor than the value of any item not
satisfying the binary condition.
In this way the process determines two regions as explained in what follows.
We will call the resulting ranges
of values {\em large} and {\em small}, respectively, where the two ranges are disjoint.
Items with large values, i.e., from the large range,
are called large, and items with small values, i.e.,  from the small range,
are called small.  Note that this definition of large and small does not indicate that the size of a large item is larger than the size of a small item (since sizes maybe be based for example on subtraction of the values).

Note that when an item is presented, its size is defined without
any knowledge of the assignment, so it is still unknown at that moment
if the
binary condition holds for this item. Thus, its value is defined
without the knowledge regarding whether it is small or large. This
knowledge is gained based on the action of the algorithm once the
item is packed. Based on the packing, if its value is required to
be large, future values will be much smaller, and if its value is
required to be small, future values will be much larger.

The construction allows us to define positive values smaller than
a given value $\eps>0$, such that for a pre-defined
(constant)
multiplicative factor $k$,
any large value is more than $k$ times larger than any small
value. Thus, there is a value $\gamma<\eps$ such that every small
value is smaller than $\frac{\gamma}k$ and every large value is
larger than $\gamma$. If items are one dimensional and their sizes
are simply these values, this means for example that an item of
size $1-\gamma$ can be packed with $k$ small items, but cannot be
packed with one large item
into the same bin.
Note that in this case large items are
also quite small, though not as small as small items. It is
possible to define items differently in one dimension, and not only in the way that their sizes are equal to the values. One option
is to use the values as complements of sizes (to $1$). Another
option is to use an additive term, for example, items can have
sizes of $\frac 13$ plus the defined value. In this case, one can
define a value of $\eps$ such that $\frac 13+\eps<\frac 12$, for
example. For items that are vectors, one can define a part of the
components to be defined based on the corresponding value. For
example, it is possible in the case $d=5$ that two components will
be equal to the value while three other components are equal to
zero.

\section{The lower bound for large values of $\boldsymbol{d}$\label{sec:large-values}}
In this section we consider the cases where the dimension is very large.
Our lower bounds are in fact not smaller than the weaker results of Azar et al. \cite{ACKS13}  which were proved only for the absolute competitive ratio measure (as a function of the dimension for large enough values of $d$), while here we consider the stronger measure of asymptotic competitive ratio.

In order to consider the asymptotic competitive measure, we have an integer parameter $N$, and both the upper bound on the cost of
an optimal offline solution as well as the lower bound on the cost of the
online solution constructed by an algorithm are linear functions of $N$.  Then, by letting $N$ grow to infinity, the lower bound on the asymptotic competitive ratio follows.

We have exactly $d$ phases in total,
where the $j$th phase consists of $N$ identical items whose size vector has $1$ as their $j$th component, all components of indexes larger than $j$ are $0$ while the components of indexes smaller than $j$ are either $0$ or $\eps$ where $0<\eps\leq \frac{1}{Nd}$.
The number of items will be $N\cdot d$, so components of values no larger than $\eps$ (zero or $\eps$) will not prevent the packing of items into the same bin. The idea of such components was presented in \cite{ACKS13}.

Note that by definition, two items of the same phase cannot fit into a common bin, and furthermore, a collection of items fit into a bin if and only if no item in the collection has an $\eps$ component if another item of the collection has $1$ in the same component.  To complete the description of the input sequence, we still need to describe the specific rules for defining components of indexes smaller than $j$ (if they are $0$ or $\eps$).
Furthermore, a construction of this form satisfies that any solution has cost of at least $N$ and of at most $Nd$ regardless of its specific details.

The definition of the item sizes depends on a fixed
offline solution that is maintained after each phase.   Our offline solution is not necessarily optimal, though we use its cost as an upper bound on the optimal cost.
The offline solution has an integer parameter $\nu$ (which will be chosen later as a function of $d$), and maintains $\nu$ classes of bins each of which has $N$ bins that we refer to as the bins of the class. One can think of the numbers $1,2,\ldots,\nu$ as colors, though there are multiple instances of bins for every color.

We say that the offline solution packs the items of the $j$th phase into class $i$ and mean that the $\ell$th item of the phase is packed into the $\ell$th bin of the class (for all $\ell=1,2,\ldots ,N$).
Using this rule means that in order to guarantee feasibility of this packing we can consider only the first bin of each class (simply because the items of each phase are the same and the packing of each bin of the class is equivalent up to indexes of equal sized items).

For every phase $j$, there is a unique class $\ell$ of our offline solution such that the items of phase $j$ are packed into bins of class $\ell$.
Thus, the assignment of phases to classes is a surjective function.
In this case, we will say that class $\ell$ is the class of phase $j$. On the other hand, every class can be the class of several phases, and for a class $\ell$, we will say that these phases are the phases of class $\ell$. Note that the phases of a certain class are defined gradually, they are an empty set initially, and for every new phase $j$ whose class is defined to be $\ell$, the set of phases of $\ell$ is extended by $j$.

Next, consider a bin $B$ of the online algorithm (at some point during the lower bound construction).  We associate with every such bin the subset $S(B)$ of classes (of the offline solution) that contain items that are packed into $B$. Thus, it will always be the case that $|S(B)|\leq \nu$.
Note that the definition is with respect to classes of the offline solution, where a class may be relevant for one phase or multiple phases.
The set $S(B)$ may be extended
later, and it is called the associated set of $B$.
Here, we are interested in the existence of at least one item of the class in this bin
(an item of a phase that belongs to this class),
and we do not distinguish between the cases of one such item or more than that.
That is, when $\ell \in S(B)$, this means that there exists at least one phase $j$ such that $\ell$ is the class of phase $j$
in our offline solution, and an item of phase $j$ is packed into $B$. Note that it is possible that for two bins $B_1, B_2$
such that $\ell \in S(B_1), S(B_2)$, the values of $j$ (where $\ell$ is the class of
 phase $j$) are distinct.

Let $[\nu]=\{ 1,2,\ldots ,\nu\}$ be the set of classes of the offline solution, and let $2^{[\nu]}$ be its power set (i.e., the set of all its subsets), and let $X=2^{[\nu]}\setminus \{ \emptyset \}$.
The goal is to capture the bin types of the online algorithm with respect to classes of the offline solution.
We will need the next assumption, which will hold for two choices of pairs of $\alpha$ and $\beta$ used in the proof.

\begin{assumption}\label{asm1}
There exists a set of subsets $\cs \subseteq X$ such that $|\cs| \geq \alpha$ and for every pair $S,S'\in \cs$ such that $S\neq S'$ we have that their symmetric difference $S \triangle S'$ satisfies $|S\triangle S'| \geq \beta$.
\end{assumption}

Next, we provide specific pairs of $\alpha$ and $\beta$ that satisfy the assumption.

\begin{lemma}
Assumption \ref{asm1} holds for $(\alpha,\beta)=(2^{\nu}-1,1)$ and if $\nu$ is sufficiently large then Assumption \ref{asm1} also holds for $(\alpha,\beta)=(2^{\nu/4}, 0.3 \nu)$.
\end{lemma}
\begin{proof}
For the first part, consider $\cs$ being the set $X$, then it has $2^{\nu}-1$ elements as required, and each pair of distinct elements represents non-equal subsets of $[\nu]$ so they differ by at least one element.
We fix the value of $\beta$ to 1 in this case.

The second part was proven by \cite{HS} who showed that if we pick a random sub-collection of subsets of $[\nu]$ with $2^{\nu/4}$ subsets each of which consisting of exactly $\nu/2$ elements of $[\nu]$ (chosen independently at random), then with some positive probability (for large enough value of $\nu$) each pair of these selected subsets satisfies the condition on their symmetric difference.  Using the probabilistic method, they were able to prove our claim (deterministically) for large enough values of $\nu$ (that they have not specified).
\end{proof}

Let $\alpha,\beta$ be a pair of positive parameters. Let $\cs$ be fixed.
Next, we define a subset in $\cs$ that {\em represents} a bin $B$ of the algorithm.  We say that $S\in \cs$ represents bin $B$ of the algorithm (and that $B$ is represented by $S$) if $|S(B)\triangle S| \leq \frac{\beta}{5}$.

\begin{lemma}\label{injection}
Every bin
of the algorithm is represented by at most one set $S\in \cs$.
\end{lemma}
\begin{proof}
Assume by contradiction that a bin $B$ is represented by two sets   $S_1,S_2\in \cs$.

Recall that $|S_1\triangle S_2| \geq \beta$ (by
Assumption \ref{asm1}), but  $|S(B)\triangle S_i| \leq \frac{\beta}{5}$ for $i=1,2$
(by the definition of representation).
First, consider the elements of $S_1\setminus S_2$.  Some of those elements belong to $S(B)$ while other do not.  Observe that $(S_1\setminus S_2)\cap S(B) \subseteq S(B)\triangle S_2$ so $|(S_1\setminus S_2)\cap S(B)| \leq \frac{\beta}{5}$. Since $(S_1\setminus S_2)\setminus S(B) \subseteq S(B)\triangle S_1$, we conclude that $|(S_1\setminus S_2)\setminus S(B)| \leq \frac{\beta}{5}$.  Therefore, $|S_1\setminus S_2| \leq \frac{2\beta}{5}$.    Similarly (by changing the roles of $S_1$ and $S_2$) we conlcude that $|S_2\setminus S_1| \leq \frac{2\beta}{5}$.  Thus, $|S_1 \triangle S_2| \leq \frac{4\beta}{5}$ contradicting our assumption on $\cs$.
\end{proof}

Note that if $\alpha=2^{\nu}-1$, then every non-empty bin $B$ is simply represented by $S(B)$,
but if $\beta\geq 5$,
then there might be bins that are not represented at all
by a set of $\cs$, while by Lemma \ref{injection}, if $B$ is represented by a set in $\cs$ it is not represented by other members of $\cs$.

For a set $S\in \cs$, we denote by $n(S)$ the number of bins (of the online algorithm) that it represents.  This value is initialized as $0$ and remains non-negative in all times.  Next, we define the vectors of a phase whose index is $j$.

Assume that after the previous phase, we have computed the values $n(S)$ for all $S\in \cs$.  If $\sum_{S\in \cs} n(S) \geq \alpha \cdot \frac{N}{2}$
after phase $j-1$ (this cannot happen for $j=1$),
we stop the construction. Since every bin of the algorithm is represented by at most one set, we conclude that the algorithm has at least $\frac{\alpha N}{2}$ bins while the offline solution has
at most $\nu\cdot N$ bins, and we get a lower bound of $\frac{\alpha}{2\nu}$.
Otherwise, we pick a set $S_j \in \cs$ such that $n(S_j)$ at the moment is below $\frac{N}{2}$.
In the second case, the existence of $S_j$ can be guaranteed by the pigeonhole principle.
The items for this phase are defined carefully in what follows.
The set $S_j$ is a subset of classes that is not represented sufficiently in the packing of the algorithm.

The items of phase $j$ (where $1 \leq j \leq d$) have a component of $\eps$ in all indexes $j'<j$ such that the class of phase $j'$ is in $[\nu]\setminus S_j$.  In the $j$th component these items have $1$, while all other components (of these items) are zero (this includes all components larger than $j$ and components smaller than $j$ such that the classes of their phases are in $S_j$).
By definition, the first phase consists of $N$ items that are unit vectors whose first component is equal to $1$, and every phase has exactly $N$ items.

The online algorithm may pack items of phase $j$ only in bins $B$ such that before the item is added it holds that $S(B) \subseteq S_j$, which we show next. This includes packing an item into a new bin, for which the current value of $S(B)$ before the item is packed is the empty set. Let $\ell$ be a class such that $\ell\in[\nu]\setminus S_j$.
If class $\ell$ is currently empty, then every bin $B$ does not contain an item of an empty class, so $\ell \notin S(B)$, as required.
Otherwise, for any bin $B$ such that $\ell \in S(B)$ there is an item with a component of $1$ for some phase $j'$ whose class is $\ell$, while every item of phase $j$ has a $j'$th component of $\eps$ for every such $j'$. In this case there is at least one such value of $j'$.

Since there are $N$ identical items (in this phase and any other phase), and the number of bins represented by $S_j$ is at most $\frac{N}{2}$, we conclude that at least $\frac{N}{2}$ of these items are packed
by the online algorithm into bins   that were not represented by $S_j$. 

If less than $d$ phases were completed, we already saw that the number of bins of the algorithm is at least $\frac{\alpha\cdot N}2$. This may happen after phase $d$ as well, but in that case we will not use this bound, since we analyze $d$ phases generally in a different way. In order to show the bound after $d$ phases, we use a potential function.  Our potential function used for assisting us in finding a lower bound on the number of bins of the algorithm is as follows. It is the sum of $|S(B)|$ over all bins $B$ of the algorithm.  That is, $\Phi = \sum_B |S(B)|$.  Observe that the cost of the online algorithm is at most $\Phi$ and not smaller than $\frac{\Phi}{\nu}$, since for any $B$ it holds that $|S(B)|\leq \nu$.

Let $\gamma=\max \{ 1, \frac{\beta}{5}\}$, which allows us to treat the two cases of $\beta$ uniformly. Consider a bin $B$ that was used to pack an item of phase $j$ by the online algorithm, and $B$ was not represented by $S_j$ (prior to this phase). In this case $S_j$ has at least $\gamma$ elements which are not in $S(B)$.  This holds since by definition of a set that represents a bin, since the cardinality of the symmetric difference between the sets is at least $\gamma$, and since $S(B) \subseteq S_j$ as was shown earlier.

Recall that there are at least $\frac N2$ items of phase $j$ packed by the algorithm into bins not represented by $S_j$, all of which packed into distinct bins. Every such bin $B$ has at least $\gamma$ elements of $S_j\setminus S(B)$. Since $|S_j|\leq \nu$, which holds for all subsets of classes, using the pigeonhole principle, we find that among the elements of $S_j$ there is at least one element that does not belong to the associated sets of at least $\frac{\gamma}{\nu} \cdot \frac{N}{2}$ bins which were used by the online algorithm for packing the items of the $j$th phase.  We pick one such element $o_j \in S_j$ (where $o_j$ is a class),
and pack the items of phase $j$ in the offline solution in class $o_j$.  This is a feasible offline packing, as the bins of class $o_j$ still have zeroes in component $j$, and the items of phase $j$ have zeroes in all components $j'$ such that the items of phase $j'$ are packed into bins of class $o_j$ (since $o_j\in S_j$).
Now, we can repeat and define the items of phase $j+1$ in the same way (if $j<d$).

Furthermore, since $o_j$ was not a member of the sets $S(B)$ for at least $\frac{\gamma}{\nu} \cdot \frac{N}{2}$ bins that the algorithm used for packing the items of phase $j$, we conclude that the value of $\Phi$ increases by at least $\frac{\gamma}{\nu} \cdot \frac{N}{2}$ while packing the $j$th phase items.  Thus, after $d$ phases the value of $\Phi$ is at least $d\cdot \frac{\gamma}{\nu} \cdot \frac{N}{2}$ unless the construction was stopped earlier with the cost of the online algorithm being $\frac{\alpha N}{2}$. The cost of the algorithm in the first case is at least $\frac{\Phi}{\nu} \geq d\cdot \frac{\gamma}{\nu^2} \cdot \frac{N}{2}$.

Recall that the cost of the offline solution is at most $\nu \cdot N$.  It remains to conclude the lower bound, where the lower bound is not smaller if we can increase the value of $\nu$ and it is maximized if $d\cdot \frac{\gamma}{\nu^2} \cdot \frac{N}{2} \approx \frac{\alpha N}{2}$.  Thus, we need a method for selecting $\nu$ if the dimension $d$ is given.

We will use a value $\nu$ for which the corresponding  pair $(\alpha,\beta)$ satisfies $d\geq \frac{\nu^2 \alpha}{\gamma}$ (where $\gamma$ is determined by $\beta$), and then the resulting lower bound would be $\frac{\alpha}{2\nu}$.

First consider the case where $\nu$ is relatively small and we use $(\alpha,\beta)=(2^{\nu}-1,1)$ and thus $\gamma =1$, and $\nu$ is an integer such that $\nu^2\cdot (2^{\nu}-1) \leq d$.  It is sufficient to require that $\nu^2\cdot 2^{\nu} \leq d$ that is satisfied by letting $\nu = \lfloor \log_2 d- 2\log_2 \log_2 d\rfloor$ as for this choice $\nu^2 2^{\nu} \leq (\log_2 d)^2 \cdot \frac{d}{(\log_2 d)^2}=d$.  The resulting lower bound is not smaller than $$\frac{\alpha}{2\nu} \geq \frac{2^{\log_2 d -2\log_2 \log_2 d}-1}{4\cdot (\log_2 d-2\log_2 \log_2 d ) }  \geq \frac{d-1}{8(\log_2 d)^3} \ , $$ where the last inequality holds for $d > 16$ as for these values of $d$ we have that $8\log_2\log_2 d < 4\log_2 d$.

Next, consider the case where the dimension is higher and we could use $(\alpha,\beta)=(2^{\nu/4}, 0.3 \nu)$ and thus $\gamma = 0.06\nu$.  We pick $\nu$ as the largest integer such that $\frac{\alpha \nu^2}{\gamma} \leq d$  that is, $\alpha \nu = \nu 2^{\nu/4}  \leq 0.06 \cdot d$.
Letting $\nu'=\frac{\nu}{4}$ we will require $\nu'\cdot 2^{\nu'} \leq 0.015 d$.  This condition is satisfied e.g. for $\nu'= \lfloor \log_2 d - \log_2 \log_2 d -7 \rfloor $ as for this choice of $\nu'$ we have $$\nu'\cdot 2^{\nu'} \leq (\log_2 d) \cdot 2^{\log_2 d - \log_2 \log_2 d -7} =  (\log_2 d) \cdot \frac{d}{\log_2 d \cdot 2^7} = \frac{d}{2^7} \leq 0.015 d$$ and $\nu'$ is integer and thus also $\nu=4\nu'$ is integer.
The resulting lower bound is $$\frac{\alpha}{2\nu}=\frac{2^{\nu'}}{8\nu'} \geq \frac{2^{\log_2 d - \log_2 \log_2 d -8}}{8\cdot ( \log_2 d - \log_2 \log_2 d -7)} \geq \frac{\frac{d}{\log_2 d}}{2^{11}\cdot \log_2 d} = \frac{d}{2^{11}\cdot (\log_2 d)^2} ,$$  that holds for large enough values of $d$.

Thus, we conclude the construction of this section by the following theorem.

\begin{theorem}
For every fixed dimension $d > 16$, there is a lower bound of $$\frac{d-1}{8(\log_2 d)^3} \ , $$ on the asymptotic competitive ratio
 of online algorithms for VP.
For sufficiently large and fixed  values of $d$ the lower bound on the asymptotic competitive ratio is $$ \frac{d}{2^{11}\cdot (\log_2 d)^2} \ .$$
\end{theorem}

\section{A lower bound for medium sized dimensions\label{sec:medd}}

Let $\alpha$, $\beta \geq 2$ be the two positive integers such that $d\geq 2+\alpha(\beta-2)$.  Next we show a lower bound of $ \frac{\alpha \cdot \beta}{\alpha+\beta -2}$ for the corresponding special case.  Note that choosing the  values $\alpha=\beta$ results in a lower bound of $\Omega(\sqrt{d})$ so for very large dimension this result is inferior to the general lower bound we considered earlier.
However, the hidden constants in the $\Omega$ notation are smaller for the current construction leading to better lower bounds for medium sized dimensions.

Let $N>d$ be a large integer such that $\frac{N}{\alpha}$ is an integer.    Our sequence of items may have up to $N^3 $ items, and we let $\eps<\frac{1}{N^3}$.  We will use the adaptive construction method to generate a sequence of
scalar
values
where $a_i$ is the value associated with the $i$th item, such that all values are smaller than $\eps$, and furthermore the following condition holds.  If an item is large, then its value is at least $N^3 $ times larger than the value of a small item.  The logical condition that we will use to define small and large items is that a
$d$-dimensional
item is large if it is packed into an empty bin and otherwise it is small. We stress the property that every item of the construction will have a
one-dimensional associated
value, and we will explain how this value is used in the definition of the
$d$-dimensional item.

During the adaptive construction, after packing the current item, there will be a value $\mmu<\frac 1N <\frac 12$. The value of $\mmu$ may decrease (but cannot increase) after assigning an item. The value $\mmu$ will satisfy the property that the value of every large item that appeared up to (and  including) the current iteration is strictly larger than $\mmu$ while for any subset of items $S$ where $S$ contains only small items that appear in the instance
(both during the prefix and later on) or large items that appear later on in the input sequence, the total value of $S$ is strictly smaller than $\mmu$.  This is obtained by letting $\mmu$ be the current upper bound on values of items that can still be either small or large at termination, and reducing the length
of the interval of possible values in the adaptive construction by a multiplicative factor of $N^3$ after each item. This is done to ensure that all such subsets $S$, whose  number of elements will be less than $N^3$ (as this will be a valid upper bound on the number of items in the entire construction), will satisfy the requirement.

Our construction will have $\beta$ phases, and it will be useful to denote by $\mmu_i$ the
current value of $\mmu$ at the end of phase $i$ (i.e., after packing the last item of phase $i$ and modifying the current interval according to the rules of the adaptive construction). Furthermore, phase $i$ uses the value $\mmu_{i-1}$.

The first and last phase have special properties while the intermediate phases ($\beta-2$ phases) are all similar.  Each phase lasts until the first point in time in which the algorithm has opened $N$ new bins during this phase.  Thus, we will ensure that the
total
cost of the algorithm is $N\cdot \beta$.  We also maintain an integer value $\pi$ denoting the {\em current component} that is being dealt with, and it has a special role in the construction. We will show later that $\pi$ will always be an index of a component (i.e., $\pi \leq d$), even though it is increased frequently.
The value $\pi$ is an index of a component such that items have a very big (and close to $1$) component of this index. By increasing $\pi$, we change (and increase) the position of this very big component in the construction.
This value is initially set as $\pi=2$ (while the very first component has a special role during the first phase), and it increases gradually,
each time by $1$, and it never decreases,
as items are being presented. We will see that the value of $\pi$ never exceeds $d$, and during the presentation of items of intermediate phases it will hold that $\pi \leq d-1$.

\paragraph{The first phase.} We construct a sequence of items, where the first component of every vector is $\frac{1}{N}$ while every other component equals to the value of the current item.  Note that this phase ends after at most $N^2$ items, since any phase ends after the algorithm used $N$ new bins, every new bin can contain at most $N$ items, and there are no bins of previous
phases which can be used.

\paragraph{The $\boldsymbol{\beta-2}$ intermediate phases.}  Every such phase will contain at most $\alpha \cdot N$ items.
We keep a counter $j$ of the index of the phase, where $j$ is initialized to $2$.  At the end of phase $j-1\geq 1$ we set $\mmu_{j-1}$
as we described above and we start presenting new items of the $j$th phase.

Each item of these $\beta-2$ intermediate phases will consist of the following components.  All components
with indexes
smaller than $\pi$ are equal to $0$, the current component
with index
$\pi$ is set to $1-\mmu_{j-1}>\frac 12$, and all other components (of larger indexes) are equal to the value of the current item (of the adaptive construction).  Recall that a new phase starts whenever the number of new bins during the current phase is $N$, and just before starting a new phase we also increase $\pi$ by $1$ (for $j=2$ $\pi$ is not increased but it is initialized).
However, there are other events where we decide to increase the value of $\pi$ by $1$.
These additional events are stated as follows.  Whenever the number of large items (according to the adaptive construction) that were packed while the value of $\pi$ is its current value, is $\frac{N}{\alpha}$, we increase the value of $\pi$ by $1$. This is done since the number of large items whose $\pi$th component is very big is the maximum possible number.
We will show that an increase in the value of $\pi$ will happen after at most $N$ consecutive items for which we used the same value of $\pi$.
This happens either due to the latter rule or due to the end of the phase (since $\pi$ is always increased due to that event).   Before presenting the vectors of the last phase, we prove the main correctness claims regarding the intermediate phases that allow our construction to have the required structure and allow us to prove the claimed lower bound on the asymptotic competitive ratio.

\begin{lemma}
The items of phase $j$
(where $2 \leq j \leq \beta-1$)
cannot be packed into bins that were opened in an earlier phase, that is, bins used first for an item of an earlier phase.  Additionally, the value of $\pi$ remains constant without being increased for at most $N$ items.
\end{lemma}
\begin{proof}
Each bin that was opened by the algorithm in the first $j-1$ phases has a large item (since by the adaptive construction, every first item of any bin of the algorithm for these phases is
large). Every large item has
at least one component that is larger than $\mmu_{j-1}$, where those components form a suffix starting with component that is the value of $\pi$ (at the time of packing the item) plus $1$. Since $\pi$ is increased by $1$ when a new phase starts, the component
indexes
corresponding to the current (larger) value of $\pi$ is larger than $\mmu_{j-1}$. Since for the items of phase $j$
there is a component equal to $1-\mmu_{j-1}$ for the current value of $\pi$, for which bins constructed
in previous phases have a sum larger than $\mmu_{j-1}$, they cannot be packed into bins of the algorithm containing items of earlier phases.
This concludes the first part.

Regarding the second claim, it is only required to consider the items of one phase, since $\pi$ is increased after every phase ends.
Assume by contradiction there are $N+1$ items that have a common value of the current component.
By the first part, these items are packed into new bins opened during the phase. Since items having the same value of $\pi$ have a $\pi$th component larger than $\frac 12$, they are packed into different bins. Thus, the phase would end before the $(N+1)$th item is presented (or it may end earlier or $\pi$ may be increased earlier), a contradiction.
\end{proof}

\begin{lemma}
During an intermediate phase, the value of $\pi$ is increased at most $\alpha$ times
 by $1$ (including the increase due to the end of the phase).
\end{lemma}
\begin{proof}
There are
exactly
$N$ large items in a phase, since every phase has $N$ new bins, and by the construction, the large items are exactly the first items of new bins.
Thus by definition the events in which we increase $\pi$ may happen at most $\alpha$ times during a phase (and in the last such event, the phase ends as well).
\end{proof}

We consider the value of $\pi$ at the beginning of the last phase.
By the last lemma, we conclude that just before
the moment when phase $\beta -1$ ends (the last intermediate phase), we have $\pi \leq 2+(\beta-2)\cdot \alpha+(\beta-1)\leq d-1$
and $\pi$ is increased to a value of at most $d$ once that phase ends. This final value  (of at most $d$)
for $\pi$ was not used as the value of $\pi$ in the definition of items of any phase (after the very last time that $\pi$ was increased, no items are defined so it was not used to define an item).

\paragraph{The last phase.}  In the last phase we present exactly $N$ identical items that are defined as follows.  In component $d$ they are equal to $1-\mmu_{\beta -1}$ and all other components are $0$.  Observe that every bin that the algorithm has opened in one of the earlier phases has one large item whose $d$th component is larger than $\mmu_{\beta -1}$, and thus the algorithm needs to open $N$ new bins for these $N$ items of the last phase.

\paragraph{Proving the resulting lower bound.}  Since there are $\beta$ phases and the algorithm is forced to open $N$ new bins in every phase, we conclude that the cost of the algorithm is exactly $N\cdot \beta$.  Since $N$ could be an arbitrary large integer, in order to prove the lower bound on the asymptotic competitive ratio of the algorithm, it suffices to show that the optimal offline cost is at most $N\cdot (1+\frac{\beta-2}{\alpha})+1$.  In order to present this proof, we will consider
all items of the last phase as small items.  We present an offline solution of cost at most $N\cdot (1+\frac{\beta-2}{\alpha})+1$.  This offline solution will pack all large items into $N\cdot (\frac{\beta-2}{\alpha})+1$ bins, and all small items into a disjoint set of $N$ bins, and in total we will use at most $N\cdot (1+\frac{\beta-2}{\alpha})+1$ bins.  First, we consider the large items.
\begin{lemma}
There is an offline solution that packs all large items into at most $N\cdot (\frac{\beta-2}{\alpha})+1$ bins.
\end{lemma}
\begin{proof}
Note that there are $\beta-1$ phases containing large items (as in the last phase all items are small).  Therefore, it suffices to show that it is possible to pack all large items of a common phase $i$ into $\frac{N}{\alpha}$ bins and to pack all large items of the first phase into one bin.

Consider the first phase.  There are $N$ large items, and one can pack all such $N$ items into a single bin because in the first phase every component of the items of the phase is at most $\frac{1}{N}$.

Next, consider the set of large items of phase $i$ for $i\geq 2$.  We pack these large items into $\frac{N}{\alpha}$ bins, as follows.  We traverse the sublist of these large items sorted by their arrival order (i.e., by the same order they had in the original input) and we pack them one by one into these bins using a round-robin approach.  That is, the large item of index $\frac{N}{\alpha} \cdot \Gamma+\gamma$ for an
integer $\Gamma$ and for $\gamma\in \{ 0,1,\ldots ,\frac{N}{\alpha}-1 \}$
will be packed into the bin of index $\gamma+1$.
Here the index of the item is its position along the sublist of large items of phase $i$, and the index of the bin is among the $\frac{N}{\alpha}$ dedicated bins for large items of this phase.  Alternatively, this packing can be seen as $\frac{N}{\alpha}$ bins, where every bin has one large item for every  value of $\pi$ used in the suitable
phase.
To see that this is a feasible packing note that in every component and every bin there is at most one item packed into the bin that has this component equals to $1-\mmu_{i-1}$ and all other items packed there have this component equal to their values.
Since the total
associated value of this set of items is smaller than $\mmu_{i-1}$ as they are all items of iterations after the iteration in which we define $\mmu_{i-1}$, these items do not exceed the bound of $1$ in every component (on their sum).
\end{proof}

Next, we consider packing of the small items into $N$ bins.
\begin{lemma}
There is an offline solution that packs all small items into $N$ bins.
\end{lemma}
\begin{proof}
We consider the sublist of (only) the small items sorted by their arrival order in the original instance.  Once again we use the round-robin approach and we pack these items into $N$ bins using round-robin.  That is, the $j$th small item (i.e., counting only small items) is packed into the bin of index $1+((j-1) \mod N )$.  It suffices to show that the resulting packing is feasible.

Consider one component in one bin of this packing and we will show that the sum of the component over all items packed into this bin is at most $1$.  For the first component, the claim holds as only the items of the first phase have nonzero first component, and since there are less than $N^2$ such (small) items, by averaging, we pack at most $N$ first-phase small items into this bin, so their total first component is at most $1$.  Next, consider the $j$th component where $j>1$.  In this component, the bin has at most one vector with $j$th component larger than $1-\eps$
(since there are at most $N$ such vectors in the instance and they appear consecutively along the input sequence, so they are assign to distinct bins by the round-robin approach).  The other vectors packed into the bin are such that their $j$th component is their
 associated value or zero, and these are small items, so they all fit into this component (no matter what is the phase in which the item with $j$th component larger than $1-\eps$
was presented).
\end{proof}

Thus, we conclude the following result.
\begin{theorem}
If $d\geq \alpha(\beta -2)+2$, then there is no online
algorithm for VP whose asymptotic competitive ratio is smaller than $\frac{\beta}{1+\frac{\beta-2}{\alpha}} = \frac{\alpha \cdot \beta}{\alpha+\beta -2}$.
\end{theorem}

For large values of $d$, we can use $\alpha=\beta = \lfloor\sqrt{d-1} \rfloor+1 \geq \lceil \sqrt{d-1} \rceil$  for which $ \alpha(\beta -2)+2 \leq (\sqrt{d-1}+1)\cdot (\sqrt{d-1}-1)+2 = d-1-1+2=d$ and this lower bound on the asymptotic competitive ratio is $ \frac{\alpha \cdot \beta}{\alpha+\beta -2} = \frac{\alpha^2}{2\alpha-2}= \frac{\alpha+1}2+\frac{1}{2(\alpha-1)}>\frac{\lfloor\sqrt{d-1} \rfloor }{2}+1 \geq \frac{\sqrt{d-1}+1}2 >\frac{\sqrt{d}}2$.  However, for small dimensions we could do better.  For example for $d=98$, we could pick $\alpha=12, \beta=10$ and the lower bound on the asymptotic competitive ratio is $\frac{120}{20} = 6$ whereas using $\alpha=\beta=10$ the lower bound is $\frac{100}{18}\approx 5.555$.

For small dimensions, namely $d=6, 7, 9, 10$ and $11$, the next estimation can be used.
Letting $\beta=3$ and $\alpha=d-2$, the lower bound is greater or equal to
$\frac{3\alpha}{\alpha+1} = \frac{3}{1+\frac{1}{\alpha}}= \frac{3}{1+\frac{1}{d-2}}$.
It gives a lower bound of
$2.4$, $2.5$,
$2.625$,
$2.666$, and
$2.7$ for the cases of $d=6,7,9,10,11$, respectively. We mention several other small values of $d$. For $d=12$, we can use $\alpha=5$ and $\beta=4$ to obtain a lower bound of $\frac {20}7 \approx 2.857$. For $d=14$,
we can use $\alpha=6$ and $\beta=4$ to obtain a lower bound of $3$. For $d=16$, we can use $\alpha=7$ and $\beta=4$ to obtain a lower bound of $\frac {28}9 \approx 3.111$.
Improved bounds for the cases $d=3,4,5,8$ are presented in the next sections.

\section{The case $\boldsymbol{d=3}$\label{sec:d=3}}
Recall that the known lower bound on the asymptotic competitive
ratio for $d=2$ is just slightly above $2$, and this was the best known constant lower bound for any small value of $d$ till now.
We prove here a lower
bound of $2.25$  for the case $d=3$,
and explain how it can be slightly improved.

We will use an adaptive construction as explained earlier. The construction is based on that of \cite{BBDEL_CCBP}.

Let $K>1000$ be a large integer, and let $\eps>0$ be a small
constant (in particular, $\eps<\frac{1}{K}<0.001$).  The input consists of three parts and we describe the parts one by one.

\paragraph{The first part of the input.}  Using an
adaptive construction of values, we define a sequence of values in
$(0,\eps)$ such that any large value is
strictly larger by a multiplicative
factor larger than $10K$ from any small value. The binary
condition is that the item is packed into an empty bin by the
online algorithm.

Thus, an item packed into an empty bin is large,
and otherwise the item is small. The number of items will be
$2\cdot K \cdot N$ for a large integer $N>0$. Letting $a_1, a_2,
\ldots, a_{2KN}$ be the sequence of values, the vector for item
$i$ is defined as follows. The first component is $\frac 1{K}$,
and each one of the two other components is equal to $a_i$. Let
$\gamma$ be a threshold such that if the
$i$th constructed value is small, it holds
that $0< a_i < \frac{\gamma}{10K}$ and otherwise $\gamma<a_i
<\eps$. The input up to this point is
denoted by $I_0$.
\begin{lemma}\label{lem10}
The optimal cost for packing $I_0$ is $2N$.
\end{lemma}
\begin{proof}
Since all first components
of all items
are equal to $\frac{1}{K}$
while the other components are smaller, an optimal solution can
pack all items into $2N$ bins and it cannot pack the items into a smaller number of bins.
\end{proof}

Let $X$ denote the number of bins used by the
algorithm for the first part of the input and by definition this is also the number of large items. Let
$\mu=\frac{\gamma}{10}$. Thus, every $K$ small values have total
value below $\mu$.

\paragraph{The second part of the input.} For the value of $\mu$ that is based on the action of the
algorithm, we define the next part of the input. There are $N$
items of each one of the two types: $(0,1-4\cdot \mu,\mu)$ and
$(0,\mu,1-4\cdot \mu)$.
So, in total there are $2N$ items of these types.
The input at this time is called $I_1$, i.e. $I_1$ is the input consisting of the first two parts of the input together.

\paragraph{Two offline packings of $\boldsymbol{I_1}$.} We
define two offline packings, for which the first part of the input
is packed in a fixed manner.  For each possibility of the third part of the input, we will use one of those offline packings that we present here. Consider the first part of the input
(the items of $I_0$), and separate small items from large items.
Large items are packed such that every bin has $K$ of them (where
one such bin may have a smaller number of these items). The large
items require $\lceil \frac XK \rceil$ bins. These are feasible bins because none of the
components of the sum of any bin
is above $1$, since no component of any item is
above $\frac{1}{K}$. Small items are also packed $K$ in each bin,
and there are at most $2N$ such bins since the total number of
items for the first part is $2KN$.
The total number of small items is $2KN-X$.
The first component of the bin
has load $1$, but the other components have loads below $\mu$.
Every such bin can also receive one item of each type of the
second part. In one packing, we partition the items of the second part into pairs where every pair consists of items of different types.  In this offline packing each pair is packed together, these pairs are first packed into bins
with $K$ small items of the first part, and if there are any
unpacked items of the second part, they are packed into new bins
(also in pairs). In the second packing, every bin gets just one
such item of the second part. For both offline packings, the numbers of bins do not exceed $\frac{X}{K}+1+2N$.

\paragraph{The third part of the input.} The third part of the input may contain two alternative sets of
items.
In the first case, leading to the input $I_{21}$, there are
$N$ items of the type $(0, 1-\mu, 1-\mu)$. Every such item is
packed into a different bin by any algorithm. In the offline packing, these items are packed first into bins with (at most) $K$ small items of the
first part (but without large items of the first part and without any
items of the second part) and then into new bins.
The online algorithm cannot combine such items with any item into the same bin, as we will see.

In the second case, leading to the input $I_{22}$, there are $N$
items of each of the types: $(0,3\mu, 1-2\mu)$, $(0,
1-2\mu,3\mu)$. No pair of such items can be packed into one bin,
but it can join a bin with (at most) $K$ small items of the first part and
one item of the second part (of the suitable type) but without large items of the first part.
It also cannot be packed with an item of the other type of the third part.

This concludes the description of the input construction.  Next, we turn our attention to proving the resulting lower bound on the asymptotic competitive ratio for the case $d=3$.

\paragraph{Proving the resulting lower bound.}
Here, we prove the following result.
\begin{theorem}\label{thm_d=3}
There is no online algorithm for the case $d=3$ whose asymptotic competitive ratio is smaller than $\frac {9}{4}$.
\end{theorem}
\begin{proof} We can assume that $X \leq 6N$, since the case where there is an
infinite number of values of $N$ for which $X \geq 6N$ implies
that the asymptotic competitive ratio of the algorithm is at least
$3$.
Recall that $OPT(I_0 )=2N$ (by Lemma \ref{lem10}).
Thus, the number of bins for the offline packing
of $I_{21}$ and the offline packing of $I_{22}$ do not
exceed $\frac{6N}{K}+1+2N$.
By the offline packing of $I_1$ and the packing of the third part.
By requiring $K \geq N$, the number of
bins is at most $2N+7$.

As for the algorithm, all bins created by the algorithm for the
first part of the input have large items, whose values are above
$\gamma=10\mu$. Thus, new bins are created for the second part and
third part. Any bin can contain at most two items of the second
part, and at most one item of the third part.

Let $Z_1$ and $Z_2$
be the numbers of bins with one and two items, respectively,
opened by the algorithm for the second part (when we count the number of items of such bins at the end of the second part). By the numbers of
items in the relevant part of the input
we have $$Z_1+2\cdot Z_2 = 2N . $$
We recall that the online algorithm cannot pack any items into bins opened during the first phase, because these bins contain
a large item of the first phase which has a second and third component at least $10\mu$ while packing a second part item
it requires at most $4\mu$ loading in these coordinates.

For $I_{21}$, every item of
the third part requires a new bin, so the cost of the algorithm is
$X+Z_1+Z_2+N$.
For $I_{22}$, it
is possible to use bins of the second part only if they have a single second part item, so
the cost of the algorithm is
at least  $X+2N+Z_2$.

Since this construction can be used for an infinite
number of values of $N$, we can use the definition of the asymptotic competitive ratio as the $\limsup$ of the absolute competitive
ratio when the optimal cost is at least $N$, and let $N$ grow to infinity. Thus, letting $R$ be the
competitive ratio, we have $$X+2N+Z_2 \leq R \cdot (2N+7)$$ and
$$X+Z_1+Z_2+N \leq R \cdot (2N+7) .$$
Taking the sum of these
inequalities gives $2X+Z_1+2\cdot Z_2 + 3N \leq 2R(2N+7)$ and by
$Z_1+2\cdot Z_2 = 2N$ and $X\geq 2N$,we have $R \geq \frac{9N}{2(2N+7)}$, which
implies a lower bound of $2.25$ on the asymptotic competitive ratio.
\end{proof}

A very slight improvement over the lower bound which we proved above in Theorem \ref{thm_d=3} can be obtained as follows, similarly to
the known construction for $d=2$ \cite{BBDEL_CCBP}. An alternative second
part of the
input will contain items whose first component is not zero but
some multiple of $\frac{1}K$.
The second component will be
slightly larger than $\frac 13$, where there will be large items
and small items (small items are those that are packed by the
algorithm into a bin that cannot receive another item), and all these items have second components larger than $\frac 13$.
There may
be a third part of the input
(in this alternative input), similarly to the construction of
\cite{BBDEL_CCBP}. We omit the details as the idea is similar and the
improvement is very small.
We note that this value of $2.25$ is a valid lower bound for the cases $d\geq 4$ as well.
For the case $d=5$ we could get the same lower bound by another method in Section \ref{sec:medd} as well, where we proved significantly larger values for larger dimensions.

\section{The case $\boldsymbol{d=8}$\label{sec:d=8}}
We consider this special case as well, in order to demonstrate
that the asymptotic competitive ratio grows relatively fast with
the dimension.  We picked the value of $d=8$ as for this dimension we are able to exhibit new properties of instances leading to improved lower bounds.  Once again the lower bound construction consists of three parts.

\paragraph{The first part of the input.} The first part of the construction is identical to that of the
case $d=3$,
including the property that the values of $K$ and $\eps$ are the same,
with the only change that the components equal to
$a_i$ are not just the second and third components, but all
components with indexes
$2, 3, \ldots, 7$ are equal to $a_i$. The first
component is still $\frac 1K$, while the $8$th component is equal
to zero. The values $\gamma$ and $\mu$ are defined as in the first part of the construction for the case $d=3$.

\paragraph{The second part of the input.} The second part of the input consists of $6N$ items consisting of $6$ groups each of which has $N$ vectors, where every
$N$ vectors of a common group are identical. All these vectors have $8$th components
equal to $\frac 13$ and first components equal to zero. The other
components are equal to either $0$ or to $1-3\mu$, where every
item has
exactly
one component equal to $1-3\mu$, and we will call it the
large component of the item.  We will have the items of group $j$ (for $j=1,2,\ldots ,6$) having component $j+1$ equal to $1-3\mu$ while all other components (excluding the $8$th component) are zero.

\paragraph{Analyzing the packing of the algorithm at the end of the second part.}  Before describing the third part of the input, we introduce some notation and properties of the packing of the algorithm at the end of the second part of the input.  The  items of the second part are defined so that no bin can have more than three such items by the constraint on the $8$th component, and all
(at most three) items of one bin have distinct large components
since $1-3\mu > \frac 12$.
We will distinguish the cases where a bin contains three, two, or just one item of the second part of the input, introducing notation for their corresponding bin numbers.

For
$j_1,j_2,j_3 \in \{2,3,4,5,6,7\}$, where $j_1<j_2<j_3$, let
$X_{j_1,j_2,j_3}$ be the number of bins with
 (exactly)
three items of the
second part of the input, whose large components are $j_1$, $j_2$,
and $j_3$. There are $20$ such variables. For $j_4,j_5 \in
\{2,3,4,5,6,7\}$, where $j_4<j_5$, let $Y_{j_4,j_5}$ be the number
of bins with
(exactly)
two items of the second part of the input, whose
large components are $j_4$ and $j_5$. There are $15$ such
variables. For $j_6 \in \{2,3,4,5,6,7\}$, let $Z_{j_6}$ be the
number of bins with
(exactly)
one item of the second part of the input,
whose large component is $j_6$. There are six such variables.
Since every bin opened by the algorithm for the first part of the
input has a sum of components of items above $\gamma=10\mu$ in components
$2,3,\ldots,7$, all these bins of the algorithm are new.

The
number of bins opened by the algorithm for the first part of the
input is denoted by $Q$ and it satisfies $$Q\geq 2N \  .$$ The sum of
variables of the form $X_{j_1,j_2,j_3}$ is denoted by $X$, the sum
of variables of the form $Y_{j_4,j_5}$ is denoted by $Y$, and the
sum of variables of the form $Z_{j_6}$ is denoted by $Z$.  That is, $X=\sum_{j_1,j_2,j_3} X_{j_1,j_2,j_3}$, $Y=\sum_{j_4,j_5} Y_{j_4,j_5}$, and $Z=\sum_{j_6} Z_{j_6}$. By counting the number of items of the second part, we have
$$3X+2Y+Z = 6N \  . $$

\paragraph{The third part of the input.}
The third part has one of ten possible sets of items, of similar
structures. These items have six non-zero components, which are
components $2, 3, \ldots,7$. Every item has three components whose
values are $2\mu$, and three components whose values are $1-\mu$.
No two such items can be packed into the same bin since the sum of such a component for two items is either $2\mu+1-\mu>1$ or $2(1-\mu)>1$.
For a triple
$\{2,j_7,j_8\}$ where $j_7,j_8 \in \{3,4,5,6,7\}$
are fixed component indexes, and
$j_7<j_8$, the input consists of $N$ items whose components $2,
j_7, j_8$ are equal to $1-\mu$ and the components
$\{2,3,4,5,6,7\}\setminus \{2,j_7,j_8\}$ are equal to $2\mu$, and
$N$ items whose components $2, j_7, j_8$ are equal to $2\mu$ and
the components $\{2,3,4,5,6,7\}\setminus \{2,j_7,j_8\}$ are equal to
$1-\mu$.  This completes the construction of the input.  Next, we prove the resulting lower bound.

\paragraph{Proving the resulting lower bound.}  Recall the decision variables $X,Y,Z,Q$ whose values are determined by the algorithm but they satisfies the conditions $Q\geq 2N$ and $3X+2Y+Z = 6N$ established above.  We first upper bound the optimal offline cost after the third part of the input and then present a lower bound on the maximum cost of the algorithm on these $10$ inputs that can be constructed in the third part of the input.
As in the proof for
$d=3$ we can assume $\frac{Q}{K}\leq 6$.

\begin{lemma}\label{lem12}
For each of the ten inputs that can be created at the end of the third part, there is an offline solution whose cost is at most $2N+7$.
\end{lemma}
\begin{proof}
Fix one particular input at the end of the third part, that is, we are given the pair $j_7,j_8$ used by the adversary for this fixed input.
An offline solution packs the items as follows. There are
 $\lceil \frac{Q}K \rceil$ bins with $K$
large items of the first part each, and $2N$ bins, where each such
bin has at most $K$ small items of the first part. Out of the last $2N$
bins, there are $N$ bins with three items of the second part whose
large components are $2, j_7, j_8$ and one item of the third part
for which these components are equal to $2\mu$, and $N$ bins with
three items of the second part whose large components are in the
set $\{2,3,4,5,6,7\}\setminus \{2,j_7,j_8\}$, and one item of the
third part for which these components are equal to $2\mu$.
These bins are feasible because of the item sizes of the input parts.
Thus, the
optimal cost never exceeds $\frac{Q}K+1+2N \leq 2N+7$.
\end{proof}

The algorithm can combine
into a common bin
some items of the third part with items
of the second part but it cannot use bins that were used for items of the first part for packing items of the second or third parts. We describe only bins without any items of the first part because any bin of the algorithm containing an item of the first part cannot receive any
additional items. We discuss such bins with two or three items of
the second part (since bins with just one item can always receive additional items).
For a set $\{2,j_7,j_8\}$, there may be bins with two items of the second
part, where one of the items has a large component in the set
$\{2,j_7,j_8\}$ and the other one has a large component in the set
$\{2,3,4,5,6,7\}\setminus \{2,j_7,j_8\}$.
In addition,
there may be  bins with three items of the second part,
where the set of
large components is none of the sets $\{2,j_7,j_8\}$ and
$\{2,3,4,5,6,7\}\setminus \{2,j_7,j_8\}$
(comparing them as sets and not as ordered tuples).

By Lemma \ref{lem12}, for large values of $N$ we find for the asymptotic competitive ratio $R$ that
$$Q+X+Y+Z \leq 2RN \ .$$

We introduce two new variables $Y',X'$ where $Y'$ is the
number of bins of the algorithm with two items of the second part
that cannot receive an item of the third part, and $X'$ is the
number of bins of the algorithm with three items of the second
part that cannot receive an item of the third part (for the choice
of third part, that is, we fix the third part temporarily).
Then, by considering this input using the fact that the third part of the input requires packing items into
at least $2N$ bins and we cannot use $Q+Y'+X'$ of the bins which were opened for the first or second parts, we conclude that $$Q+Y'+X'+2N \leq 2RN \ .$$
We take the sum of the last inequality for all
ten options of $j_7$, $j_8$, where the right hand side is $20RN$,
and the multiplier of $N$ on the left hand side is $20N$.
Since we consider all options for the third part of the input, the values $X'$ and $Y'$ can have different values, and more precisely, each one has up to ten different values.

We count
the multiplier of each variable as follows. Variables of the form
$X_{j_1,j_2,j_3}$ are included in all variables $X'$ except for
the option where $j_1,j_2,j_3$ are components of equal values,
(the algorithm chose exactly the same subset as the one chosen for the third part of the input)
which is just one case of the third part. Thus, the multiplier of
$X_{j_1,j_2,j_3}$  is $9$.  Variables of the form $Y_{j_4,j_5}$
are included in all variables $Y'$ except for cases where $j_4$
and $j_5$ are components of equal values, which is the case if
$\{j_4,j_5\}\subseteq \{2,j_7,j_8\}$ or  $\{j_4,j_5\}\subseteq
\{2,3,4,5,6,7\}\setminus \{2,j_7,j_8\}$. Out of the $15$ variables,
there are nine such options that are included (in the sense that the bins cannot be used for items of the third part of the input) and six that are not
included. Thus, every variable is included in six of the ten partitions, and in this sum of constraints every variable $Y_{j_4,j_5}$ has a
multiplier of $6$. We get
\begin{equation} \label{d=8eq}  10Q+6Y+9X+20N \leq
20RN \ . \end{equation}
Using $3X+2Y+Z = 6N$ and $Q+X+Y+Z \leq 2RN$, we find by
subtraction that $2X+Y -Q \geq 6N-2RN$ or alternatively $$9X+4.5 Y
-4.5 Q \geq 27N-9RN \ . $$ By subtracting the last inequality from
\eqref{d=8eq}, we have $1.5Y+14.5Q + 47 N \leq 29RN$.
Since $Y \geq 0$ and $Q \geq 2N$ hold, we establish that $76 N \leq 29 R N$
and therefore $R \geq \frac{76}{29} \approx 2.620689655$ as we summarize in the following theorem.

\begin{theorem}
There is no online algorithm for the case $d=8$ whose asymptotic competitive ratio is smaller than  $\frac{76}{29} \approx 2.620689655$.
\end{theorem}

\bibliographystyle{abbrv}

\end{document}